\documentclass{article}

\usepackage{fullpage}
\usepackage{amsmath}
\usepackage{graphicx}
\usepackage{rotating}
\usepackage{amssymb}
\usepackage{color}
\usepackage{amsmath}
\usepackage{amsthm}
\usepackage{mathtools}
\usepackage{cite}

\newtheorem{theorem}{Theorem}
\newtheorem{lemma}[theorem]{Lemma}

\newtheorem{corollary}[theorem]{Corollary}

\DeclareMathOperator*{\argmin}{\arg\!\min}

\allowdisplaybreaks

\title{The Power and Limitations of Static Binary Search Trees with Lazy Finger}

\date{}

\author{
Prosenjit Bose\thanks{Carleton University, Ottawa, Ontario, Canada.} 
\and Karim Dou\"{i}eb\thanks{Brussels, Belgium.}
\and John Iacono\thanks{Polytechnic Institute of New York University, Brooklyn, New York, USA. Research supported by NSF Grant CCF-1018370.}
\and Stefan Langerman\thanks{Universit\'e Libre de Bruxelles, Brussels, Belgium. Ma\^{i}tre de Recherches du F.R.S.-FNRS.}
}

\begin{document}

\maketitle

\begin{abstract}
A static binary search tree where every search starts from where the previous one ends (\emph{lazy finger}) is considered. Such a search method is more powerful than that of the classic optimal static trees, where every search starts from the root (\emph{root finger}), and less powerful than when rotations are allowed---where finding the best rotation based tree is the topic of the dynamic optimality conjecture of Sleator and Tarjan. The runtime of the classic root-finger tree  can be expressed in terms of the entropy of the 
distribution of the searches, but we show that this is not the case for the optimal lazy finger tree. A non-entropy based asymptotically-tight expression for the runtime of the optimal lazy finger trees is derived, and a dynamic programming-based method is presented to compute the optimal tree.
\end{abstract}

\section{Introduction}

\subsection{Static trees}

A binary search tree is one of the most fundamental data structures in computer science. In response to a search operation, some binary trees perform changes in the data structure, while others do not. For example, the splay tree \cite{DBLP:journals/jacm/SleatorT85} data structure performs a sequence of searches that moves the searched item to the root. Other binary search tree data structures do not change at all during a search, for example, red-black trees \cite{DBLP:conf/focs/GuibasS78} and AVL trees \cite{tree-adelson-62}. We will call BSTs that do not perform changes in the structure during searches to be \emph{static} and call trees that perform changes \emph{BSTs with rotations}. In this work we do not consider insertions and deletions, only searches, and thus can assume without loss of generality that all structures under consideration are storing the integers from 1 to $n$ and that all searches are to these items.

We consider two variants of static BSTs: root finger and lazy finger. In the classic method, the \emph{root finger} method, the first search proceeds from the root to the item being searched. In the second and subsequent searches, a root finger BST executes the searches in the same manner, always starting each search from the root. In contrast, here we consider \emph{lazy finger} BSTs to be those which start each search at the destination of the previous search and move to the item being searched. In general, this movement involves going up to the least common ancestor (LCA) of the previous and current items being searched, and then moving down from the LCA to the current item being searched.

\subsection{Notation and definitions}

A static tree $T$ is a fixed binary search tree containing $n$ elements. No rotations are allowed. The data structure must process a sequence of searches, by moving a single pointer in the tree. 
Let $r(T,i,j)$ be the time to move the pointer in the tree $T$ from node $i$ to $j$. If $d_T(i)$ represents the depth of node $i$, with the root defined as having depth zero, then 
\begin{align*}
r(T,i,j)&=d_T(i)-d_T(\text{LCA}_T(i,j)) + d_T(j)- d_T(\text{LCA}_T(i,j))\\&=d_T(i)+d_T(j)-2d_T(\text{LCA}_T(i,j)).
\end{align*}

The runtime to execute a sequence $X=x_1,x_2, \ldots x_m$ of searches on a tree $T$ using the root finger method is
$$
R_{r}(T,X) = \sum_{i=1}^m r(T,root(T),x_i)=\sum_{i=1}^m d_T(x_i)
$$
and the runtime to execute the same sequence on a tree $T$ using the lazy finger method is
$$ R_{\ell}(T,X) = \sum_{i=1}^m r(T,x_{i-1},x_i)
=\left(2\sum_{i=1}^m (d_T(x_i)-d_T(\text{LCA}_T(x_i,x_{i-1})))\right)-d_T(x_m)$$
where $x_0$ is defined to be the root of $T$, which is where the first search starts.

\subsection{History of optimal static trees with root finger}

For the root finger method, once the tree $T$ is fixed, the cost of any single search in tree $T$ depends only on the search and the tree, not on any of the search history. Thus, the optimal search tree for the root finger method is a function only of the frequency of the searches for each item.
Let $f_X(a)$ denote the number of searches in $X$ to $a$. Given $f_X$, computing the optimal static BST with root finger has a long history.
In 1971, Knuth gave a $O(n^2)$ dynamic programming solution that finds the optimum tree \cite{DBLP:journals/acta/Knuth71}.
More interestingly is the discovery of a connection between the runtime of the optimal tree and the entropy of the frequencies:

$$H(f_X) = \sum_{a=1}^{n} \frac{f_X(a)}{m}\lg \frac{m}{f_X(a)}$$.

Melhorn \cite{DBLP:journals/acta/Mehlhorn75}
showed that a simple greedy heuristic proposed by Knuth
\cite{DBLP:journals/acta/Knuth71}
and shown to have a linear-time implementation by Fredman
 \cite{DBLP:conf/stoc/Fredman75}
 produced a static tree where an average search took time
 $2+\frac{1}{1-\lg (\sqrt{5}-1)}H(f_X)$. Furthermore, Melhorn demonstrated
 a lower bound of $\frac{1}{\lg 3} H(f_X)$ for an average search in an optimal static tree, and showed this bound was tight for infinitely many distributions. Thus, by 1975, it was established that the runtime for an average search in an optimal search tree with root finger was $O(H(f_X))$, and that such a tree could easily be computed in linear time.

\subsection{Our results}

We wish to study the natural problem of what we have coined search with a lazy finger in a static tree, i.e. have each search start where the last one ended. We seek to characterize the optimal tree for this search strategy, and describe how to build it.

The lazy finger method is asymptotically clearly no worse then the root finger method; moving up to the LCA and back down is better than moving to the root and back down, which is exactly double the cost of the root finger method. But, in general, is the lazy finger method better?
For the lazy finger method, the cost of a single search in a static tree depends only on the current search and the previous search---this puts lazy finger's runtime dependence on the search sequence between that of root finger and trees with rotations. Thus the optimal search tree for the lazy finger method only depends on the frequency of each search transition; let $f_X(a,b)$ be the number of searches in $X$ to $b$ where the previous search was to $a$. Given these pairwise frequencies (from which the frequencies $f_X(a)$ can easily be computed), is there a nice closed form for the runtime of the optimal BST with lazy finger? One natural runtime to consider is the conditional entropy:

$$H_c(f_X) = \sum_{a=1}^{n} \sum_{b=1}^{n} \frac{f_X({a,b})}{m}\lg \frac{f_X(a)}{f_X({a,b})}$$

This is of interest as information theory gives this as an expected lower bound\footnote{When multiplied by $\frac{1}{\lg 3}$, as the information theory lower bound holds for binary decisions and as observed in \cite{DBLP:journals/acta/Mehlhorn75} needs to be adjusted to the ternary decisions that occur at each node when traversing a BST.} if the search sequence is derived from a Markov chain where $n$ states represents searching each item.

While a runtime related to the conditional entropy is the best achievable by any algorithm parameterized solely on the pairwise frequencies, however, we will show in Lemma~\ref{l:conditionalentropysuck} that the conditional entropy is impossible to be asymptotically achieved for any BST, static or dynamic, within any $o(\log n)$ factor. Thus, for the root finger, the lower bound given  by information theory is achievable, for lazy finger it is not related in any minimal way with the runtime of the optimal tree. In Section~\ref{s:multitree} we will present a simple static non-tree structure whose runtime is related to the conditional entropy.

\newcommand{\sumw}[3]{\sum_{#1=\min(#2,{#3})}^{\max({#2},{#3})} w_{#1}}
\newcommand{\sumwt}[3]{\sum_{#1=\min(#2,{#3})}^{\max({#2},{#3})} w^T_{#1}}

This still leaves us with the question: is there a simple closed form for the runtime of the optimal BST with lazy finger? We answer this in the affirmative by showing an equivalence between the runtime of BSTs with lazy finger and something known as the weighted dynamic finger runtime. In the weighted dynamic finger runtime, if item $i$ is given weight $w_i$, then the time to execute search $x_i$ is
$$\lg  \frac{ \sumw{k}{x_i}{x_{i-1}}}{\min(w_{x_{i-1}},w_{x_i})}.$$
Our main theorem is that the runtime of the best static tree with lazy finger, $LF(X)$, is given by the weighted dynamic finger runtime bound with the best choice of weights:
$$LF(X)=
\min_T R_\ell(T,X) =
\Theta \left(\min_W \left\{ \sum_{i=1}^m \lg  \frac{\displaystyle \sumw{k}{x_i}{x_{i-1}}}{\min(w_{x_{i-1}},w_{x_i})} \right\} \right)
$$

To prove this, we first state the result of Seidel and Aragon \cite{DBLP:journals/algorithmica/SeidelA96} in Section~\ref{s:weightsgivetrees} of how to construct a tree with the weighted dynamic finger runtime given a set of weights. Then, in Section~\ref{s:trresgiveweights}, we show how, given any static tree $T$, there exists weights such the runtime of $T$ on a sequence using lazy finger can be lower bounded using the weighted dynamic finger runtime with these weights. These results are combined in Section~\ref{s:main} to give the main theorem.

While a nice closed-form formula for the runtime of splay trees is not known, there are several different bounds on their runtime: working set, static finger, dynamic finger, and static optimality \cite{DBLP:journals/jacm/SleatorT85,DBLP:journals/siamcomp/ColeMSS00,DBLP:journals/siamcomp/Cole00}. One implication of our result is that the runtime of the optimal lazy finger tree is asymptotically as good as that of all of the aforementioned bounds with the exception of the working set bound (see Theorem~\ref{th:limit} for why the working set bound does not hold on a lazy finger static structure).

However, while these results have served to characterize the best runtime for the optimal BST, a concrete method is needed to compute the best tree given the pairwise frequencies. We present a dynamic programming solution in Section~\ref{s:dynprog}; this solution takes time $O(n^3)$ to compute the optimal tree for lazy finger, given a table of size $n^2$ with the frequency of each pair of searches occurring adjacently. 
This method could be extended using the ideas of \cite{DBLP:journals/ijcga/IaconoM12} into one which periodically rebuilds the static structure using the observed frequencies so far; the result would be an online structure that for sufficiently long search sequences achieves a runtime that is within a constant factor of the optimal tree without needing to be initialized with the pairwise frequencies.

%

\subsection{Why static trees?}

Static trees are less powerful than dynamic ones in terms of the classes of search sequence distributions that can be executed quickly, so why are we studying them? Here we list a few reasons:

\begin{itemize}
\item Rotation-based trees have horrible cache performance. However, there are methods to map the nodes of a static tree to memory so as to have optimal performance in the disk-access model and cache-oblivious models of the memory hierarchy \cite{DBLP:journals/ipl/Boas77,DBLP:journals/corr/cs-DS-0410048,DBLP:journals/jal/GilI99,DBLP:conf/soda/ClarkM96}. One leading cache oblivious predecessor query data structure that supports insertion and deletion works by having a static tree and moves the data around in the fixed static tree in response to insertions and deletions and only periodically rebuilds the static structure \cite{DBLP:journals/jal/BenderDIW04}---in such a structure an efficient static structure is the key to obtaining good performance even with insertions and deletions.

\item One should use the simplest structure with the least overhead that gets the job done. By completely categorizing the runtime of the optimal tree with lazy finger, one can know if such a structure is appropriate for a particular application or whether one should instead use the more powerful dynamic trees, or simpler root-finger trees.

\item Concurrency becomes a real issue in dynamic trees, which requires another layer of complexity to resolve (see, for example \cite{DBLP:conf/ppopp/BronsonCCO10}), while static trees trivially support 
concurrent operations. Moreover, if several search sequences from several sources are interleaved, any dependence the pervious operation is destroyed by the interleaving. However, it is easy to have each sequence have it own lazy finger into a static tree, that allows each search sequence to be executed concurrently while not losing the ability to take advantage of any distributional temporal cohesion in each search sequence source.

\end{itemize}

\section{Weights give a tree} \label{s:weightsgivetrees}

\begin{theorem}[Seidel and Aragon \cite{DBLP:conf/focs/AragonS89}]
Given a set of weights $W=w_1, w_2, \ldots w_n$,  there is a randomized method to choose a tree $T_W$ such that the expected runtime is

$$ r(T_W,i,j) = O\left(\lg  \frac{ \sumw{k}{i}{j}}{\min(w_i,w_j)} \right) $$
\end{theorem}

The method to randomly create $T_W$ is a straightforward random tree construction using the weights: recursively pick the root using the normalized weights of all nodes as probabilities.
Thus, by the probabilistic method \cite{DBLP:books/wi/AlonS92}, there is a deterministic tree, call it $T_{\text{minw}}(X)$ whose runtime over the sequence $X$ is 
at most the runtime bound of Seidel and Agraon for the sequence $X$ on the best possible choice of weights.

\begin{corollary}\label{c:tree}

There is a tree $T_{\text{\text{minw}}}(X)$ such that
$$
\sum_{i=1}^m r(T_{\text{\text{minw}}}(X),x_{i-1},x_i) =
O \left(\min_W \left\{ \sum_{i=1}^m \lg  \frac{ \sumw{k}{x_i}{x_{i-1}}}{\min(w_{x_{i-1}},w_{x_i})} \right\} \right)
$$
\end{corollary}

\begin{proof}
This follows directly from Seidel and Aragon, where $T_{\text{\text{minw}}}(X)$ is a tree that achieves the expected runtime of their randomized method for the best choice of weights.
\end{proof}

\section{Trees can be represented by weights}\label{s:trresgiveweights}

\begin{lemma}
For all trees $T$ there is a set of weights $W^T=w_1^T,w_2^T,\ldots w_n^T$ such that for all $i,j$

$$r(T,i,j) \leq \left(\lg  \frac{ \sumwt{k}{i}{j}}{\min(w_i^T,w_j^T)} \right) $$
\end{lemma} \label{treeweight}

\begin{proof}
These weights are simple: give a node at depth $d$ in $T$ a weight of $\frac{1}{4^d}$.
%
%
%
Consider a search that starts at node $i$ and goes to node $j$. Such a path goes up from $i$ to $\text{LCA}_T(i,j)$ and down to $j$. A lower bound on $\sumwt{k}{i}{j}$ 
is the weight of $\text{LCA}_T(i,j)$ which is included in this sum and is $\frac{1}{4^{d_T(\text{LCA}_T(i,j))}}$. Thus we can bound $\lg  \frac{ \sumwt{k}{i}{j}}{\min(w_i^T,w_j^T)}$ as follows:

\begin{align*}
\lg  \frac{ \sumwt{k}{i}{j}}{\min(w_i^T,w_j^T)}
&\geq  \lg  \frac{\frac{1}{4^{d_T(\text{LCA}_T(i,j))}}}{\min \left(\frac{1}{4^{d_T(i)}},\frac{1}{4^{d_T(j)}}\right)}\\
&=2\max(d_T(i),d_T(j))-2d_T(\text{LCA}_T(i,j))\\
& \geq  d_T(i)+d_T(j)-2d_T(\text{LCA}_T(i,j))\\
& =2r(T,i,j)
\end{align*}

\end{proof}

\begin{corollary} \label{c:weights} For any tree $T$
$$\sum_{i=1}^m r(T,x_{i-1},x_i)  =
\Omega \left(\min_W \left\{ \sum_{i=1}^m \lg  \frac{ \sumw{k}{x_i}{x_{i-1}}}{\min(w_{x_{i-1}},w_{x_i})} \right\} \right)$$.
\end{corollary} 

\begin{proof}
Follows directly from Lemma~\ref{treeweight},  summing over the access sequence $X$, and noting that replacing the weights $W^T$ with minimum weights can only decrease the right hand side.
\end{proof}

\section{Proof of main theorem}\label{s:main}

Here we combine the results of the previous two sections to show that the runtime of the optimal tree with lazy finger, $LF(T)$, is asymptotically the weighted dynamic finger bound for the best choice of weights.

\begin{theorem}
$$
LF(T)=
\min_T \left( \sum_{i=1}^m r(T,x_{i-1},x_i) \right) =
\Theta \left(\min_W \left\{ \sum_{i=1}^m \lg  \frac{\sumw{k}{x_i}{x_{i-1}}}{\min(w_{x_{i-1}},w_{x_i})} \right\} \right)
$$
\end{theorem}

\begin{proof}
We start by defining $T_{\min}(X)$ to be the optimal tree.
Let 
$$T_{\min}=\argmin_T \left\{ \sum_{i=1}^m r(T,x_{i-1},x_i) \right\}.$$
Then, the runtime of $T_{\min}$ is at most the runtime of any other tree, including $T_{\text{minw}}$:  \begin{equation} \sum_{i=1}^m r(T_{\min},x_{i-1},x_i) \leq \sum_{i=1}^m r(T_{\text{minw}},x_{i-1},x_i). \label{eq:mins} \end{equation} 

We now proceed with the main proof. By Corollary~\ref{c:tree}:
\begin{align*}
\min_W \left\{ \sum_{i=1}^m \lg  \frac{ \sumw{k}{x_i}{x_{i-1}}}{\min(w_{x_{i-1}},w_{x_i})} \right\} 
&= \Omega \left( \sum_{i=1}^m r(T_{\text{minw}}(X),x_{i-1},x_i) \right) \\
\intertext{By Equation~\eqref{eq:mins}}
&= \Omega  \left( \sum_{i=1}^m r(T_{\min}(X),x_{i-1},x_i)  \right)\\
\intertext{Using Corollary~\ref{c:weights}:}
&=
\Omega \left(\min_W \left\{ \sum_{i=1}^m \lg  \frac{ \sumw{k}{x_i}{x_{i-1}}}{\min(w_{x_{i-1}},w_{x_i})} \right\} \right)
\end{align*}
\end{proof}

\section{Hierarchy and limitations of models}

In this section we show there is a strict hierarchy of runtimes from the root finger static BST model to the lazy finger static BST model to the rotation-based BST model. Let $OPT(X)$ be the fastest any binary search with rotations  can execute $X$.

\begin{theorem} \label{th:limit}
For any sequence $X$, 
$$\min_T R_r(T,X) = \Omega\left( \min_T R_{\ell}(T,X) \right) =\Omega(OPT(X)). $$
Furthermore there exist classes of search sequences of any length $m$, $X^1_m$ and $X^2_m$ such that
$$\min_T R_r(T,X^1_m) = \omega \left( \min_T R_{\ell}(T,X^1_m) \right)$$
and
$$ \min_T R_{\ell}(T,X^2_m)  =\omega(OPT(X^2_m)).$$
\end{theorem}

\begin{proof}
We address each of the claims of this theorem separately.

\paragraph{Root finger can be simulated with lazy finger: $\min_T R_r(T,X) = \Omega\left( \min_T R_{\ell}(T,X) \right)$.} For lazy finger, moving up to the LCA and back down is not more work than than moving to the root and back down, which is exactly double the cost of the root finger method.

\paragraph{Lazy finger can be simulated with a rotation-based tree:
 $\min_T R_{\ell}(T,X) =\Omega(OPT(X))$.}
 The normal definition of a tree allowing rotations has a finger that starts at the root at every operation and can move around the tree performing rotations.
The work of \cite{icalp} shows how to simulate with constant-factor overhead any number of lazy fingers in a tree that allows rotations in the normal tree with rotations and one single pointer that starts at the root. This transformation can be used on a static tree with lazy finger to get the result.

\paragraph{Some sequences can be executed quickly with lazy finger but not with root finger: There is a $X^1_m$ such that $\min_T R_r(T,X^1_m) = \omega \left( \min_T R_{\ell}(T,X^1_m) \right)$.} One choice of $X^1_m$ is the sequential search sequence $1,2,\ldots n,1,2, \ldots$ repeated until a search sequence of length $m$ is created. So long a $m\geq n$, this takes time $O(m)$ to execute on any tree using lazy finger, but takes $\Omega(m \lg n)$ time to execute on every tree using root finger.

\paragraph{Some sequences can be executed quickly using a BST with rotations, but not with lazy finger.} \label{l:conditionalentropysuck}
Pick some small $k$, say $k=\lg n$. Create the sequence $X^2_m$ in rounds as follows: In each round pick $k$ random elements from $1..n$, search each of them once, and then perform $n$ random searches on these $k$ elements. Continue with more rounds until at total of $m$ searches are performed. A  splay tree can perform this in time $O(m \lg k)$. This is because splay trees have the working-set bound, which states that the amortized time to search an item is at most big-O of the logarithm of the number of different things searched since the last time that item was searched. For the sequence $X^2_m$, the $n$ random searches in each round have been constructed to have a working set bound of $O(\lg k)$ amortized, while the $k$ other searches in each round have a working set bound of $O(\lg n)$ amortized. Thus the total cost to execute $X^2_m$ on a splay tree is 
$ O\left( \frac{m}{n+k}(n \lg k + k \lg n) \right) $
which is $O(m \lg \lg n)$ since $k=\lg n$.

However, for a static tree with lazy finger,  $X^2_m$ is basically indistinguishable from a random sequence and takes $\Omega(m \lg n)$ time.
This is because the majority of the searches are random searches where the previous item was a random search, and in any static tree the expected distance between two random items is $\Omega(\lg n)$.
\end{proof}

\begin{lemma}
A BST in any model cannot reach the conditional entropy. 
\end{lemma}
\begin{proof}
Wilber~\cite{DBLP:journals/siamcomp/Wilber89} proved that the bit reversal sequence is performed in $\Omega(n\lg n)$ time in an optimal dynamic BST. This sequence is a precise permutation of all elements in the tree. 
However, any single permutation repeated over and over has a conditional entropy of 0, since every search is completely determined by the previous one.
\end{proof}

\section{Constructing the optimal lazy finger BST}\label{s:dynprog}

Recall that  $f_{a,b} =f_X({a,b})$ is the number of searches in $X$ where the current search is to $b$ and the previous search is to $a$, and $f_X{(a)}$ is the number of searches to $a$ in $X$.
We will first describe one method to compute the cost to execute $X$ on some tree $T$. Suppose the nodes in $[a,b]$ constitute the nodes of some subtree of $T$, call it $T_{a,b}$ and denote the root of the subtree as $r(T_{a,b})$.  We now present a recursive formula for computing the expected cost of a single search in $T$. Let $R_{\ell}(T,X,a,b)$ be the number of edges traversed in $T_{a,b}$ when executing $X$.
Thus, $R_{\ell}(T,X,1,n)$ equals the runtime $R_{\ell}(T,X)$. There is recursive formula for $R_{\ell}(T,X,a,b)$:
$$R_{\ell}(T,X,a,b)=
\begin{cases}
0\text{ if }b<a
\\

\begin{array}{l}

\overbrace{R_{\ell}(T,X,a,r(T_{a,b})-1)}^{(a)}
+
\overbrace{R_{\ell}(T,X,r(T_{a,b})+1,b)}^{(b)}
\\  \;\;\;\;\;
+
\overbrace{2 
\sum_{\mathclap{\substack{i \in [a,r(T_{a,b})-1]\\ j \in [r(T_{a,b})+1,b]}}}
 (f_{i,j}+f_{j,i})}^{(c)}
\\  \;\;\;\;\;
+\overbrace{\sum_{i\not = r} (f_{i,r(T_{a,b})}+f_{r(T_{a,b}),i})}^{(d)}
+
\overbrace{\sum_{\mathclap{\substack{i\in[a,b] \\ i \not = r(T_{a,b})\\ j\not \in [a,b]}}} (f_{i,j}+f_{j,i})}^{(e)}
\end{array}
\text{ otherwise}
\end{cases}
$$
The formula is long but straightforward. First we recursively include 
the number of edges traversed in the left (a) and right (b) subtrees of the root $r(T_{a,b})$.
Thus, all that is left to account for is traversing the edges between the root of the subtree and its up to two children. Both edges to its children are traversed when a search moves from the left to right subtree of $r_{a,b}$ or vice-versa (c). A single edge to a child of the $r(T_{a,b})$ traversed if a search moves from either the left or right subtrees of $r(T_{a,b})$ to $r(T_{a,b})$ itself or vice-versa (d), or if a search moves from any node but the root in the current subtree containing the nodes $[a,b]$ out to the rest of $T$ or vice-versa (e).

This formula can easily be adjusted into one to determine the optimal cost over all trees---since at each step the only dependence on the tree was is root of the current subtree, the minimum can be obtained by trying all possible roots.
Here is the resultant recursive formulation for the minimum number of edges traversed in and among all subtrees containing $[a,b]$:
$$\min_T R_{\ell}(T,X,a,b)=
\begin{cases}
0\text{ if }b<a
\\

\min_{r \in [a,b]} \left\{ 
\begin{array}{l}

\min_T R_{\ell}(T,X,a,r-1)
\\ 
+\min_T R_{\ell}(T,X,r+1,b)
\\
+
{ 2 \displaystyle
\sum_{\mathclap{\substack{i \in [a,r-1]\\ j \in [r+1,b]}}}
 (f_{i,j}+f_{j,i})}
\\+{ \displaystyle \sum_{i\not = r} (f_{i,r}+f_{r,i})}
\\+
{ \displaystyle \sum_{\mathclap{\substack{i\in[a,b] \\ i \not = r\\ j\not \in [a,b]}}} (f_{i,j}+f_{j,i})}
\end{array}
\right\} \text{ otherwise}
\end{cases}
$$

This formula can trivially be evaluated using dynamic programming in $O(n^5)$ time as there are $O(n^3)$ choices for $a$, $b$, and $r$ and evaluating the summations in the brute-force way takes time $O(n^2)$. The dynamic programming gives not only the cost of the best tree, but the minimum roots chosen at each step gives the tree itself. The runtime can be improved to $O(n^3)$ by observing that when $f$ is viewed as a 2-D array, each of the sums is simply a constant number of partial sum queries on the array $f$, each of which can be answered in $O(1)$ time after $O(n^2)$ preprocessing. (The folklore method of doing this is to store all the 2-D partial sums from the origin; a generic partial sum can be computed from these with a constant number of additions and subtractions).

We summarize this result in the following theorem:
\begin{theorem}
Given the pairwise frequencies $f_X$ finding the tree that minimizes the execution time of search sequence $X$ using lazy finger takes  time $O(n^3)$.
\end{theorem}

This algorithm computes an optimal tree, and takes time liner in the size of the frequency table $f$. Computing $f$ from $X$ can be done in $O(m)$ time, for a total runtime of $O(m+n^3)$. It remains open if there is any approach to speed up the computation of the optimal tree, or an approximation thereof. 
Note that although our closed form expression of the asymptotic runtime of the best tree was stated in terms of an optimal choice of weights, the dynamic program presented here in no way attempts to compute these weights. It would be interesting if some weight-based method were to be discovered.

\section{Multiple trees structure} \label{s:multitree}
Here we present a static data structure in the comparison model on a pointer machine that guarantees an average search time of $O(H_c(f_X) \log_d n )$ for any fixed value $1\leq d\leq n$, a runtime which we have shown to be impossible for any BST algorithm, static or dynamic. This data structure requires $O(dn)$ space.  In particular, setting $d=n^\epsilon$ gives a search time of  $O(H_c(f_X))$ with space $O(n^{1+\epsilon})$ for any $\epsilon>0$. 

As a first attempt, a structure could be made of $n$ binary search trees $T_1, T_2, \ldots T_n$ where each tree $T_i$ is an optimal static tree given the previous search was to $i$. By using tree $T_{x_{i-1}}$ to execute search $T_i$, the asymptotic conditional entropy can be easily obtained. However the space of this structure is $O(n^2)$. Thus space can be reduced by observing the nodes not near the root of every tree are being executed slowly and thus need not be stored in every tree.

The \emph{multiple trees structure} has two main parts. It is composed first by a complete binary search tree $T'$ containing all of $S$. Thus the height of $T'$ is $O(\lg n)$. The second part is $n$ binary search trees $\{T_{1},T_{2},\ldots,T_{n}\}$.  A tree $T_{i}$ contains the $d$ elements $j$ that have the greatest frequencies $f_X(i,j)$; these are the $j$ elements most frequently searched after that $i$ has been searched. 
The depth of an element $j$ in $T_{i}$ is $O(\lg{f_X(i) \over f_X(i,j)})$. For each element $j$ in the entire structure we add a pointer linking $j$ to the root of $T_j$.
The tree $T'$ uses $O(n)$ space and every tree $T_{j}$ uses $O(d)$ space. Thus the space used by the entire structure is $O(dn)$.

Suppose we have just searched the element $i$ and our finger search is located on the root of $T_i$. Now we proceed to the next search to the element $j$ in the following way: Search $j$ in $T_i$. If $j$ is in $T_i$ then we are done, otherwise search $j$ in $T'$. After we found $j$ either in $T_j$ or $T'$ we move the finger to the root of $T_j$ by following the aforementioned pointer. 

If $j$ is in $T_i$ then it is found in time $O(\lg{f_X(i) \over f_X(i,j)})$. Otherwise if $y$ is found in $T'$, then it is found in $O(\lg n)$ time. We know that if $y$ is not in $T_x$ this means that optimally it requires $\Omega(\lg d)$ comparisons to be found since $T_x$ contains the $d$ elements that have the greatest probability to be searched after that $x$ has been accessed. Hence every search is at most $O(\lg n/\lg d)$ times the optimal search time of $O(\lg{f_X(i) \over f_X(i,j)})$. Thus a search for $x_i$ in $X$ takes time
$O\left( \log_d n  \lg{f_X(x_i) \over f_X(x_{i-1},x_i)} \right).$
Summing this up over all $m$ searches $x_i$ in $X$ gives the runtime to execute $X$:
\begin{align*}
O\left( \sum_{i=1}^m \log_d n  \lg{f_X(x_i) \over f_X(x_{i-1},x_i)} \right)
&=
O\left( \sum_{a=1}^n \sum_{b=1}^n f_X(a,b) \log_d n  \lg{f_X(a) \over f_X(x_{a},x_b)} \right)
\\
&=O( m H_c(f_X) \log_d n)
\end{align*}
We summarize this result in the following theorem:
\begin{theorem}
Given the pairwise frequencies $f_X$ and a constant $d$, $1\leq d \leq n$, the multiple trees structure
executes $X$ in time $O( m H_c(f_X) \log_d n)$ and uses space $O(nd)$.
\end{theorem}

We conjecture that no pointer-model structure has space $O(n)$ and search cost $O( H_c(f_X))$.


\bibliographystyle{plain}
\bibliography{bib,avl}

\begin{thebibliography}{10}

\bibitem{tree-adelson-62}
G.~M. Adelson-Velskij and E.~M. Landis.
\newblock {An Algorithm for the Organization of Information}.
\newblock {\em Doklady Akademii Nauk USSR}, 146(2):263--266, 1962.

\bibitem{DBLP:books/wi/AlonS92}
Noga Alon and Joel Spencer.
\newblock {\em The Probabilistic Method}.
\newblock John Wiley, 1992.

\bibitem{DBLP:conf/focs/AragonS89}
Cecilia~R. Aragon and Raimund Seidel.
\newblock Randomized search trees.
\newblock In {\em FOCS}, pages 540--545. IEEE Computer Society, 1989.

\bibitem{DBLP:journals/jal/BenderDIW04}
Michael~A. Bender, Ziyang Duan, John Iacono, and Jing Wu.
\newblock A locality-preserving cache-oblivious dynamic dictionary.
\newblock {\em J. Algorithms}, 53(2):115--136, 2004.

\bibitem{DBLP:conf/ppopp/BronsonCCO10}
Nathan~Grasso Bronson, Jared Casper, Hassan Chafi, and Kunle Olukotun.
\newblock A practical concurrent binary search tree.
\newblock In R.~Govindarajan, David~A. Padua, and Mary~W. Hall, editors, {\em
  PPOPP}, pages 257--268. ACM, 2010.

\bibitem{DBLP:conf/soda/ClarkM96}
David~R. Clark and J.~Ian Munro.
\newblock Efficient suffix trees on secondary storage (extended abstract).
\newblock In {\'E}va Tardos, editor, {\em SODA}, pages 383--391. ACM/SIAM,
  1996.

\bibitem{DBLP:journals/siamcomp/Cole00}
Richard Cole.
\newblock On the dynamic finger conjecture for splay trees. part ii: The proof.
\newblock {\em SIAM J. Comput.}, 30(1):44--85, 2000.

\bibitem{DBLP:journals/siamcomp/ColeMSS00}
Richard Cole, Bud Mishra, Jeanette~P. Schmidt, and Alan Siegel.
\newblock On the dynamic finger conjecture for splay trees. part i: Splay
  sorting log n-block sequences.
\newblock {\em SIAM J. Comput.}, 30(1):1--43, 2000.

\bibitem{DBLP:journals/corr/cs-DS-0410048}
Erik~D. Demaine, John Iacono, and Stefan Langerman.
\newblock Worst-case optimal tree layout in a memory hierarchy.
\newblock {\em CoRR}, cs.DS/0410048, 2004.

\bibitem{icalp}
Erik~D. Demaine, John Iacono, Stefan Langerman, and Ozkur Ozkan.
\newblock Combining binary search trees.
\newblock In {\em International Colloquium on Automata, Languages and
  Programming (ICALP)}. To Appear, 2013.

\bibitem{DBLP:conf/stoc/Fredman75}
Michael~L. Fredman.
\newblock Two applications of a probabilistic search technique: Sorting x + y
  and building balanced search trees.
\newblock In William~C. Rounds, Nancy Martin, Jack~W. Carlyle, and Michael~A.
  Harrison, editors, {\em STOC}, pages 240--244. ACM, 1975.

\bibitem{DBLP:journals/jal/GilI99}
Joseph Gil and Alon Itai.
\newblock How to pack trees.
\newblock {\em J. Algorithms}, 32(2):108--132, 1999.

\bibitem{DBLP:conf/focs/GuibasS78}
Leonidas~J. Guibas and Robert Sedgewick.
\newblock A dichromatic framework for balanced trees.
\newblock In {\em FOCS}, pages 8--21. IEEE Computer Society, 1978.

\bibitem{DBLP:journals/ijcga/IaconoM12}
John Iacono and Wolfgang Mulzer.
\newblock A static optimality transformation with applications to planar point
  location.
\newblock {\em Int. J. Comput. Geometry Appl.}, 22(4):327--340, 2012.

\bibitem{DBLP:journals/acta/Knuth71}
Donald~E. Knuth.
\newblock Optimum binary search trees.
\newblock {\em Acta Inf.}, 1:14--25, 1971.

\bibitem{DBLP:journals/acta/Mehlhorn75}
Kurt Mehlhorn.
\newblock Nearly optimal binary search trees.
\newblock {\em Acta Inf.}, 5:287--295, 1975.

\bibitem{DBLP:journals/algorithmica/SeidelA96}
Raimund Seidel and Cecilia~R. Aragon.
\newblock Randomized search trees.
\newblock {\em Algorithmica}, 16(4/5):464--497, 1996.

\bibitem{DBLP:journals/jacm/SleatorT85}
Daniel~Dominic Sleator and Robert~Endre Tarjan.
\newblock Self-adjusting binary search trees.
\newblock {\em J. ACM}, 32(3):652--686, 1985.

\bibitem{DBLP:journals/ipl/Boas77}
Peter van Emde~Boas.
\newblock Preserving order in a forest in less than logarithmic time and linear
  space.
\newblock {\em Inf. Process. Lett.}, 6(3):80--82, 1977.

\bibitem{DBLP:journals/siamcomp/Wilber89}
Robert~E. Wilber.
\newblock Lower bounds for accessing binary search trees with rotations.
\newblock {\em SIAM J. Comput.}, 18(1):56--67, 1989.

\end{thebibliography}

\end{document}